\documentclass[11pt]{article}

 \usepackage[margin=25.4mm]{geometry}
\usepackage{hyperref}

\usepackage{amsfonts}
\usepackage{amsmath,amssymb}
\usepackage{graphicx,amscd,xypic}
\usepackage{bbding}
\usepackage{indentfirst}
\usepackage{mathrsfs, dsfont, bbm}
\usepackage[T1]{fontenc}

\newtheorem{thm}{Theorem}[section]
\newtheorem{defn}{Definition}[section]
\newtheorem{prop}[thm]{Proposition}
\newtheorem{coro}[thm]{Corollary}
\newtheorem{lemma}[thm]{Lemma}
\newtheorem{assump}{Assumption}[section]
\newtheorem{example}{Example}[section]
\newtheorem{remark}{Remark}
\newenvironment{proof}{\hspace{0ex}\textsc{Proof}.\hspace{1ex}}{\hfill$\Box$\\[2ex] }

\DeclareMathOperator{\BE}{\mathbf{E}}
\DeclareMathOperator{\BP}{\mathbf{P}}
\DeclareMathOperator{\PP}{\mathcal{P}}

\DeclareMathOperator{\R}{\mathbb{R}}
\DeclareMathOperator{\Rin}{\mathfrak{R}}

\DeclareMathOperator{\bx}{X}
\DeclareMathOperator{\mgf}{\mathfrak{M}}
\DeclareMathOperator{\seta}{\mathcal{A}}
\DeclareMathOperator{\setb}{\mathcal{B}}
\DeclareMathOperator{\setc}{\mathcal{C}}
\DeclareMathOperator{\setd}{\mathcal{D}}
\DeclareMathOperator{\setl}{\mathcal{L}}
\DeclareMathOperator{\lone }{\mathcal{L}^1}
\DeclareMathOperator{\sets}{\mathfrak{S}}
\DeclareMathOperator{\vp}{V_{\PP}}

\DeclareMathOperator*{\defin}{=\!\!\!=}
\newcommand{\dindex}{Duality Index}
\newcommand{\blemma}{Basic Lemma}

\begin{document}
\title{Investment under Duality Risk Measure\thanks{The author thanks Ching, Wai-Ki at the Department of Mathematics in the University of Hong Kong for his kind referring of the celebrated work of Aumann and Serrano (2008). This author acknowledges financial support from Hong Kong Early Career Scheme (No. 533112), Hong Kong General Research Fund (No. 529711) and Hong Kong Polytechnic University.}
}
\date{}
\author{Zuo Quan Xu\thanks{Department of Applied Mathematics, Hong Kong Polytechnic University, Kowloon, Hong Kong. Email: \href{mailto:maxu@polyu.edu.hk}{maxu@polyu.edu.hk}. }}
\date{12 March 2013 }
\maketitle
\begin{abstract}
One index satisfies the duality axiom if one agent, who is uniformly more risk-averse than another, accepts a gamble,  the latter accepts any less risky gamble under the index. Aumann and Serrano (2008) show that only one  index  defined for so-called gambles satisfies the duality and positive homogeneity axioms. We call it a duality index.  This paper extends the definition of duality index to all outcomes including all gambles, and considers a portfolio selection problem in a complete market, in which the agent's target is to minimize the index of the utility of the relative investment outcome. By linking this problem to a series of Merton's optimum consumption-like problems, the optimal solution is explicitly derived. It is shown that if the prior benchmark level is too high (which can be verified), then the investment risk will be beyond \emph{any} agent's risk tolerance. If the benchmark level is reasonable, then the optimal solution will be the same as that of one of the Merton's series problems, but with a particular value of absolute risk aversion, which is given by an explicit algebraic equation as a part of the optimal solution. According to our result, it is riskier to achieve the same surplus profit in a stable market than in a less-stable market, which is consistent with the common financial intuition.
\\[3mm]
\end{abstract}

\section{Introduction}
\noindent
Diamond and Stiglitz (1974) point out that whether or not a person takes a gamble depends on two distinct considerations:
\begin{itemize}
 \item[(i)] The attributes of the gamble and, in particular, how risky it is; and
 \item[(ii)] The attributes of the person and, in particular, how averse he or she is to risk.
\end{itemize}
In terms of the first issue, the concept of the risk measure has been used to explain how risky a gamble is. Many well-studied risk measures are described in the literature, such as the superhedging price, value at risk, tail value at risk, and expected shortfall as well as general coherent risk measures. These measures emphasize certain aspects of risk. However, few of them directly reflect the risk-averse person's attitude; that is, the perspective that ``risk is what risk-averters hate'' (Machina and Rothschild 2008). The entropic risk measure, which depends on such a risk aversion through the exponential utility function, is one of the few to have attempted to capture this feature.
In order to overcome the drawbacks of the existing measures, Aumann and Serrano (2008) have developed one risk measure which emphasizes such a risk-averters' attitude. This preserves many properties of the coherent risk measure such as first-order monotonicity, convexity and positive homogeneity. Unlike the coherent risk measure, however, it is also second-order monotonic, which is consistent with the emphasis on the risk-averters' attitude.
Unfortunately, Aumann and Serrano (2008) only define the measure for a certain type of discrete random variables called gambles. It is acknowledged that most outcomes in financial applications are of continuous or mixed type, so their measure cannot be applied to many of theses outcomes. To incorporate general outcomes such as price of stocks, options and general contingent claims,   this paper   generalizes the definition of the measure to cover all random variables. The measure, like the original, will satisfy an essential axiom, namely the duality axiom. This axiom states that if one agent, who is uniformly more risk-averse than another agent, accepts a gamble,   the latter will accept any less-risky gamble under the measure. It clearly demonstrates the solid connection between the measure and the attitude of the risk-averter. We therefore label it as the duality risk measure or duality index.
The axiomatic characterization of the measure will be considered in detail in the following section.
\par
In terms of the second consideration, utility functions have been used to describe the risk-aversion of an agent. The most widely used utility functions are concave, which represents that the agent is globally risk-averse. Kahneman and Tversky (1979, 1992) consider $S$-shaped utility functions\footnote{$S$-shaped utility function is convex on the negative return and concave on the positive return.} to reflect the risk-seeking attitude of the agent in a loss situation and the risk-averse attitude in a gain situation.  Meanwhile, they also introduce a reference point to separate gain and loss situations. Though many other utility functions are   considered in the literature,  only the globally risk-averse (which includes risk-neutral) agent   as well as a reference point will be discussed in this paper. The reference point reflects the agent's relative financial situation.
\par
To incorporate both considerations, namely how risky an outcome is and how risk-averse the agent is, we introduce a portfolio selection problem. This  problem aims to find out a portfolio that minimizes the duality risk measure of the utility of the relative investment outcome, that is the difference between the investment outcome and the benchmark level. The risk measure addresses the first consideration and the utility function the second. Since the duality risk measure is highly nonlinear, we adopt a novel idea to deal with the portfolio selection problem, by firstly linking the problem to a series of Merton's optimum consumption-like problems, and then solving them using the well-known Lagrange method. It turns out that the original problem is equivalent to one of the series problems but with a particular choice of absolute risk aversion, which is given by an explicit algebraic equation as a part of the explicit optimal solution. Thus the explicit solution of the original problem is derived and the problem is completely solved.
A critical threshold is also derived, so that once the surplus level (that is, the difference between the benchmark level and initial endowment) is beyond a threshold, the investment risk will exceed the agent's risk tolerance. In particular, if the agent is risk-neutral, that is to say with a linear utility function, then the investment risk will grow linearly with respect to (w.r.t.) the surplus level. The investment risk is also positively related to the entropy of the pricing kernel of the market. The result verifies the common  financial intuition  that it is much harder and riskier to achieve the same surplus profit in a stable market than in a less-stable market.
\par
The paper is organized as follows. Section 2 defines the duality risk measure for all outcomes,   studies its properties and axiomatic characterization, and then shows that it is the unique nontrivial index satisfying two axioms. Section 3 presents a portfolio selection problem under a complete market setting. The problem is to find a possible outcome to minimize the duality risk measure of the utility of the relative investment outcome. Section 4 is devoted to solving this portfolio selection problem. We first study how well posed the problem is, that is, whether its value is finite. Then we link it to a series of Merton's optimum consumption-like problems via a bridge problem. The series is then treated using the standard Lagrange method. Finally, the optimal solution and value of the original problem are derived. Analytical and numerical examples are also presented in section 4 to illustrate the main result of this paper. We  conclude the paper in section 5.

\section{Definition and Characterization of the {\dindex} }
\noindent
In order to define the duality risk measure, we need to review some of the concepts used in Aumann and Serrano (2008).
\par
A gamble is a random variable whose mean is positive and that takes finitely many values, some of which are negative.
\par
Say that an agent with utility function $u$ accepts a gamble $g$ at wealth $w$\footnote{Throughout this paper, wealth is constant.} if $\BE[u(w+g)]>u(w)$, where $\BE$ stands for ``expectation''; that is to say, the agent prefers taking that gamble to refusing it at wealth $w$.\footnote{By this definition, no agent accepts 0, which is inconsistent with financial intuition that no loss is an acceptable situation. In Aumann and Serrano (2008), it is not an issue because 0 is not regarded as a gamble. However, we will regard 0 as an outcome in this paper, so we assume that all agents accept 0 throughout this paper. It would not make any difference if 0 was assumed to be accepted by nobody.} Throughout this paper, we only consider the risk-averse (which includes risk-neutral) agent; that is to say, $u$ is concave.
\par
Say that one agent is \textit{uniformly more risk-averse} than another, if whenever the former accepts a gamble at some wealth, the latter accepts that gamble at any wealth, but not vice versa.
\par
A \emph{risk measure} or \emph{index} is a (positive) real-valued function on gambles. A gamble is less risky than another  under an index if its index value is strictly less than that of the latter.
\par
Now we will introduce two important axioms related to indices.
\begin{description}
 \item[Duality Axiom:] If one agent, who is \textit{uniformly more risk-averse than} another, accepts a gamble, then the latter agent will accept any less-risky gamble under the index.
 \item[Positive Homogeneity Axiom:]
 If a gamble is scaled by some positive scalar, then the index value is also scaled by the same scalar.
\end{description}
Aumann and Serrano (2008) show that, up to a positive multiple, there is a unique index satisfying the above two axioms. The duality axiom is more central than the other because together with the weak conditions of continuity and monotonicity it already implies that the index is unique up to the ordinal equivalent. Thus, we call the unique index satisfying both the duality and positive homogeneity axioms the \emph{Aumann-Serrano Duality Risk Measure} or \emph{Aumann-Serrano Duality Index}, or simply the \emph{Duality Risk Measure} or \emph{{\dindex}}. Some important properties of the {\dindex} are listed as follows.
\begin{description}

 \item \emph{Sub-additive}: The {\dindex} of the sum of two gambles is no more than the sum of the indices of each gamble.
 \item \emph{Law-invariant}: The Duality Indices of two identically distributed gambles are the same.
 \item\emph{Convex}: If a gamble is a linear combination of two gambles, then its {\dindex} is no more than the same combination of the indices of each gamble.
 \item \emph{Monotonic}: The {\dindex} decreases monotonically w.r.t. the first- and second-order (stochastic) dominance. \footnote{Say that one gamble \textit{first-order dominates} another one, if its value is always no less than the latter. Say that one gamble \textit{second-order dominates} another one, if the latter can be obtained by replacing some of the former's value with an outcome whose mean is that value. Say that one gamble \textit{ stochastically dominates} another one if there is a gamble distributed like the former that dominates the latter. A gamble $g$ second-order stochastic dominates another gamble $h$ if and only if $\BE[f(g)]\leqslant \BE[f(h)]$ for all decreasing and convex utility functions $f$.}
\end{description}
It is noted that the convexity property is not stated explicitly by Aumann and Serrano (2008), however,  this property will play a very important role in our analysis. It accords with the widely accepted financial wisdom that diversified investment reduces risk.
\par
Now, let us define the {\dindex} for general outcomes.

\subsection{New Definition of Duality Index}
\noindent
This paper is going to investigate  a portfolio selection problem under the {\dindex}.
Portfolio selection means finding the best possible outcome in a certain set under a certain meaning, in the current setting, that is related to the {\dindex}.
Every outcome with a mean (which may not be finite) will be considered. The definition of the {\dindex} in Aumann and Serrano (2008) cannot be employed, because the outcomes considered there only take finitely many values, and the index is not well-defined for many of  which need to be addressed here. The details will be stated after Example \ref{ex:definition} below.
Therefore, it is necessary to extend the definition of {\dindex} to cover general outcomes.
\par
Throughout this paper, the term market refers to a given probability space $(\Omega, \mathcal{F},\BP)$, and outcomes refer to random variables with well-defined finite or infinity mean values in the market. Denote the set of all outcomes as $$\setl\defin\{X:\textrm{ $X$ is a real-valued $\mathcal{F}$-measurable random variable, }\BE[X^-]<+\infty \textrm{ or }\BE[X^+]<+\infty \},$$
where $X^-=\max\{-X,0\}$ and $X^+=\max\{X,0\}$ stand for the loss part and gain part of an outcome $X$, respectively, and $\BE$ stands for the mathematical expectation on the given probability space $(\Omega, \mathcal{F},\BP)$.
\par
Define the set of outcomes with finite moment generating functions on loss as
\begin{align}\label{mgf1}
\mgf \defin \{X\in\setl:\BE[e^{ \varepsilon X^-}]<+\infty \textrm{ for some scalar $ \varepsilon>0$}\}.
\end{align}
It is clear that every moment of loss $X^-$ is finite when $X\in\mgf$ and every lower bounded outcome belongs to $\mgf$. In particular, all the gambles considered in Aumann and Serrano (2008) belong to $\mgf$.
\par
We first present several useful results.\footnote{Most of proofs in this paper are given in the Appendix.}
\begin{lemma}\label{losspress}
Fix a scalar $\varepsilon>0$. Then $\BE[e^{ -\varepsilon X}]<+\infty$ if and only if $\BE[e^{ \varepsilon X^-}]<+\infty$.
\end{lemma}
\begin{coro}\label{losspress2}
If $\BE[e^{ -\varepsilon X}]<+\infty$ for some scalar $ \varepsilon>0$, then $\BE[e^{ -\alpha X}]<+\infty$ whenever $0\leqslant \alpha\leqslant\varepsilon$.
\end{coro}
\begin{coro}\label{mgfpress}
 The set $\mgf$ can be expressed as
\begin{align}\label{mgf2}
\mgf = \{X\in\setl:\BE[e^{ -\varepsilon X}]<+\infty \textrm{ for some scalar $ \varepsilon>0$}\}
\end{align}
and
\begin{align}\label{mgf3}
\mgf &= \{X\in\setl:\textrm{ there is $ \varepsilon>0$ such that $\BE[e^{ -\alpha X}]<+\infty$ whenever $0\leqslant \alpha \leqslant \varepsilon$ }\}.
\end{align}
Moreover, $\mgf$ is a convex set.
\end{coro}
\par
Each outcome is classified into one of the following categories:
\begin{align*}
\seta & \defin \{X\in \mgf :\BE[X^+]>\BE[X^-]>0 \;\},\\
\setb & \defin \{X\in\setl :\BE[X^-]=0\;\},\\
\setc & \defin \{X\in\setl :\BE[X^-]\geqslant \BE[X^+], \;\BE[X^-]>0 \;\},\\
\setd & \defin \{X\in\setl : X\notin \seta\cup \setb\cup\setc \;\}.
\end{align*}
Note that the set $\seta\cup \setb= \{X\in \mgf :\;\BE[X]>0\;\}\cup\{0\}$ is convex. This simplifies our analysis.
\par
Before formally defining the {\dindex}, we need a very important lemma, {\blemma}, which will be used frequently in the following analysis.
\begin{lemma}[{\blemma}]
Let $\hat{\alpha}=\sup\{\alpha\geqslant 0: \BE[e^{-\alpha X}]\leqslant 1\}$ for each $X\in\setl$. Then the mapping $\alpha\mapsto \BE[e^{-{\alpha} X}]$ is continuous on $[0, \hat{\alpha}]$ and $\BE[e^{-\hat{\alpha} X}]\leqslant 1$. If $\BP(X\neq 0)>0$, then $\BE[e^{-\alpha X}]<1$ whenever $0<\alpha<\hat{\alpha} $, and $\BE[e^{-\alpha X}]>1$ whenever $\alpha>\hat{\alpha} $. Moreover, $0<\hat{\alpha}<+\infty$ whenever $X\in\seta$, $\hat{\alpha}=+\infty$ whenever $X\in\setb$, and $ \hat{\alpha}=0$ whenever $X\in\setc\cup \setd $.
\end{lemma}
It is very important to notice that $\BE[e^{-\hat{\alpha} X}]<1$ may  happen in {\blemma}.\footnote{In fact, $\BE[e^{-\hat{\alpha} X}]$ may take any value between 0 and 1.}
\begin{example}\label{ex:definition}
Let $X$ be a discrete outcome with distribution $\BP(X=-n)=n^{-2}e^{-3n-3}$ for each positive integer $n$, $\BP(X=3)=1-\sum_{n=1}^{\infty} n^{-2}e^{-3n-3}.$ Then $X\in\seta$, $\hat{\alpha}=3$, and $\BE[e^{-\alpha X}]<1$ whenever $0<\alpha\leqslant 3$ and $\BE[e^{-\alpha X}]=+\infty$ whenever $\alpha>3$.
\end{example}
It is not hard to prove that $\BE[e^{-\hat{\alpha} X}]=1$ holds in {\blemma} if and only if there is $ \varepsilon>0$ such that $1\leqslant \BE[e^{ -\varepsilon X}]<+\infty$ (for example, $X\in\seta$ is lower bounded).
This allows Aumann and Serrano (2008) to define the {\dindex} as the unique positive root of $\BE[e^{- X/{R}}]=1$.\footnote{When $X$ is a gamble defined in Aumann and Serrano (2008), the equation $\BE[e^{- X/{R}}]=1$ admits a unique positive root.} However,  as noted above, $\BE[e^{- X/{R}}]=1$ may not admit a positive solution, so we have to change the definition of the {\dindex}.
\par
Now, we define the $\emph{{\dindex}} $ for every outcome in $\setl$.
\begin{defn}\label{define:uindex}
 {{\dindex}} of $X\in\setl $ is defined as $\hat{\alpha}^{-1}$, where $\hat{\alpha}=\sup\{\alpha\geqslant 0: \BE[e^{-\alpha X}]\leqslant 1\}$.\footnote{Here $0^{-1} $ stands for $+\infty$. If the definition of {{\dindex}} is replaced by $\hat{\alpha}'^{-1}$, where $\hat{\alpha}'=\sup\{\alpha\geqslant 0: \BE[e^{-\alpha X}]< 1\}$. Then no agent would accept 0 since its {{\dindex}} would be $+\infty$. This would be consistent with the definition of acceptable in Aumann and Serrano (2008). We think, however, it would contradict the financial intuition that no loss is always an acceptable situation, so we do not adopt this definition. }
\end{defn}
This definition coincides with the original {\dindex} presented for every outcome discussed in Aumann and Serrano (2008). Denote by $\Rin(X)$ the {{\dindex}} of $X $.
\par
Let us look at the {{\dindex}} for outcome in the different categories.
\begin{itemize}
 \item If $X\in \seta $, then $0<\Rin(X)<+\infty$. The risk is moderate. Whether this outcome is accepted by an agent depends on his/her risk tolerance.
 \item If $X\in \setb $, then $\hat{\alpha}=+\infty$ and $\Rin(X)=0$. There is no risk at all. It is consistent with financial intuition. Any agent will accept this outcome because there is no potential loss.
 \item If $X\in\setc$, then $\hat{\alpha}=0$ and $\Rin(X)=+\infty$. The risk is intolerable.
 It is also in conformity financial intuition because no risk-averse agent will accept this outcome.
 \item If $X\in\setd$, then $\hat{\alpha}=0$ and $\Rin(X)=+\infty$. This risk is also intolerable.
 Because $\setd\subseteq\{X\in\setl : \BE[e^{ \varepsilon X^-}]=+\infty \textrm{ for every $ \varepsilon>0$}\;\}$, the loss is too large.
\end{itemize}
Therefore, it can be concluded that$\Rin(X)$ is finite if and only if $X$ belongs to $\seta\cup\setb$.
\par
Several essential properties of the {\dindex} are listed as follows.
\begin{prop}\label{indexproperties}
The {\dindex} has the following properties.
\begin{enumerate}
 \item The {\dindex} is subadditive: $\Rin(X+Y)\leqslant \Rin(X)+\Rin(Y)$ for any $X, Y$.
 \item The {\dindex} is positive homogeneous: $\Rin(k X)=k \Rin(X)$ for any scalar $k>0$ and $X$.
 \item The {\dindex} is convex.
 \item The {\dindex} is nonnegative and law-invariant.
 \item The {\dindex} is monotonically decreasing w.r.t. the first- and second-order (stochastic) dominance.
\end{enumerate}
\end{prop} 

\subsection{Characterization and Uniqueness}
\noindent
In this section, we show that the {\dindex} satisfies the duality and positive homogeneity axioms, which is unique.
\par
We need a technical result as follows. It shows that the {\dindex} is continuous under the meaning given below.
\begin{lemma}\label{lemma:chara}
Let $u(\cdot)$ be an increasing function and $X\in\setl$ satisfy $ u(w+X)\in\setl$ at some wealth $w$. Let $X_n=\min\{\max\{X,-n\},n\}$. Then $\lim\limits_{n\to+\infty } \Rin(X_n)=\Rin(X)$ and $\lim\limits_{n\to+\infty } \BE[u(w+X_n)]=\BE[u(w+X)]$.
\end{lemma} 
\begin{thm}\label{indexunique}
The {\dindex} defined in Definition \ref{define:uindex} is the unique, nontrivial, nonnegative real-valued index on $\setl$ that satisfies both the duality and positive homogeneity axioms.\footnote{Say that a nonnegative real-valued index is nontrivial, the index of an outcome is zero if and only if the outcome is accepted by all risk-averse agents. In fact, the latter means the outcome is nonnegative. }
\end{thm} 

\section{A Portfolio Selection Problem}
\noindent
Having introduced the {\dindex} and studied its properties.  We will consider a portfolio selection problem under the {\dindex}  in this section.
\par
An agent in the market is going to find an outcome $X$ to
\begin{align}\label{obj0}
 \min\limits_{\bx } \; \Rin(u(X -\ell ))\qquad
 \mathrm{s.t.} \quad \BE[\rho X ]=x,\quad X\in\lone,
\end{align}
where the utility function $u:\R\mapsto\R$ of the agent is concave, strictly increasing, and differentiable with $u(0)=0$; the random variable $\rho>0$ is the \emph{market price of risk} or \emph{pricing kernel} with mean value $1$; the constant $\ell$ is the \emph{benchmark} or \emph{reference point} of the agent,\footnote{For simplicity, we consider deterministic benchmark only. There is difficulty to generalize to random case.} which is typically larger than the \emph{initial endowment} $x$; and the set of possible outcomes is defined as $\setl^1 \defin\{X\in\setl:\BE[|X|]<+\infty\}$.
\par
We assume that $\BE[\rho\ln(\rho)]<+\infty$ and $\BE[ \ln(\rho)]>-\infty$.\footnote{In the Black-Scholes, or Merton models, $\rho$ is lognormal distributed. In this case, this assumption holds true.} Note that $\BE[\rho\ln(\rho)]>\BE[\rho]\ln(\BE[\rho])=0$ and $\BE[ \ln(\rho)]< \ln(\BE[\rho])=0$ by the Jensen's inequality.\footnote{Throughout this paper, we assume $\rho$ is not a constant if not elsewhere specified. }
\par
The constraint $\BE[\rho \bx ]=x$ is often called the \emph{budget constraint} and limits the choice of a particular hedging strategy as attainable or replicable given an initial endowment $x$. In this paper, we assume the market is complete, that is to say, any outcome $X\in\lone$ can be replicated. However, we do not study here how to replicate an outcome. Interested readers can consult Pardoux and Peng (1990) for an analysis of related topics.
\par
It may be observed that people   tend to think of possible outcomes relative to a benchmark or reference point rather than the absolute outcomes themselves. This phenomenon is called the framing effect. So we consider the relative outcome $ X -\ell $ instead of the absolute outcome $X$ in the target. It turns out that the benchmark plays an essential role in the problem.
\par
A risk-averse agent   tends to derive less utility from consuming additional units of the same product.
This phenomenon is called the law of diminishing marginal utility and is reflected by the concavity of the utility function in the target. If the agent is risk-neutral, then the utility function is just the identical function, and we will discuss it as an example later.
 
\section{Solving the Portfolio Selection Problem}
\noindent
Let $Y=X -\ell $. Then we can rewrite problem \eqref{obj0} as
\begin{align}\label{obj1}
 \inf\limits_{Y} \; \Rin(u(Y))\qquad
 \mathrm{s.t.} \quad \BE[\rho Y]=-y, \quad Y\in \lone,
\end{align}
where $y=\ell-x$ is called the \emph{surplus level}. We investigate this problem instead of problem \eqref{obj0}. The optimal value of problem \eqref{obj1} is defined as $V(y)$. It is clear that $V(y)=0$ whenever $y\leqslant 0$ because $Y=-y\geqslant 0$ is a feasible solution and $\Rin(u(Y))=0$. In this case, the agent's benchmark is too low, so he/she can achieve it without any risk at all.
\par
From now on, we focus on the case $y>0$.
\par
\subsection{Well-posedness}
\noindent
Let us first study the well-posedness issue (whether or not the value is finite) of problem \eqref{obj1}.
\begin{prop}\label{vfinite}
The value of problem \eqref{obj1} is finite if and only if the set
 \begin{align}\label{define:sety}
\sets_y & \defin \{Y\in \lone: \BE[\rho Y]=-y,\quad u(Y)\in\seta \;\}
\end{align}
is not empty. Moreover, if the value of problem \eqref{obj1} is finite, then the optimal solution, if it exists, must belong to $\sets_y$.
 \end{prop}
\begin{proof}
``$\Longrightarrow$'':
The value is finite, so there is $Y\in \lone$ satisfying $\BE[\rho Y]=-y$, $u(Y)\in\seta\cup\setb$. If $u(Y)\in\setb$, then $Y\geqslant 0$ a.s., $\BE[\rho Y]\geqslant 0> -y$, a contradiction. This also confirms that the set $\sets_y $ is not empty and the optimal solution, if it exists, must belong to $\sets_y$. 
\par
``$\Longleftarrow$'': This is evident.
\end{proof}
Now we study the set $\sets_y$. The condition $u(Y)\in\seta$ is not so easy to justify, so we want to find a more easily justified replacement. Let us consider the following set
 \begin{align}
\hat{\sets}_y & \defin \{Y\in \lone: \BE[\rho Y]=-y,\quad \BE[u(Y)]>0\;\}.
\end{align}
The set $\hat{\sets}_y$ is not so hard to analyze. It is related to the classic Merton's optimum consumption model (Merton 1971).
\begin{lemma}\label{setdemp} 
Define
\begin{align}\label{define:haty}
\hat{y}\defin \sup\{y: \textrm{ there is $Y\in \lone$ satisfying $\BE[\rho Y]=-y$ and $\BE[u(Y)]>0 $}\}.
\end{align}
Then the set $\hat{\sets}_y$ is not empty if and only if $y<\hat{y}$.
\end{lemma} 
It is clear that $\hat{y}\geqslant 0$ because whenever $y<0$, $Y=-y\in \lone$ satisfying $\BE[\rho Y]=-y$ and $\BE[u(Y)]>0 $. The value of $\hat{y}$ is not so hard to derive in general cases.
Let us look at two most important and widely used cases.
\begin{example}
If $u(\cdot)$ is linear, then $\hat{y}=+\infty$.\footnote{Here we assume that $\mathrm{essinf}\:\rho=0$. For example, $\rho$ is lognormal distributed. In general case, $\hat{y}=1/\mathrm{essinf}\:\rho$.}
\end{example}
\begin{example}
If $u(\cdot)$ is strictly concave, then $\hat{y}= -\lim\limits_{\lambda \to\hat{\lambda} +}\BE\left[ \rho (u')^{-1}\big(\lambda\rho\big) \right]$, where
 \begin{align}
\hat{\lambda}=\sup\left\{\lambda>0: \BE\left[u\Big((u')^{-1}\big(\lambda \rho\big)\Big)\right]>0 \right\}.
\end{align}
\end{example}
Next, we go back to study the set $ \sets_y$.
\begin{lemma} \label{setdemp2}
The set $ \sets_y$ is not empty if and only if $y<\hat{y}$.
\end{lemma} 
Putting all of the results obtained thus far together, we can solve the well-posedness issue of problem \eqref{obj1}  completely.
\begin{thm}\label{wellpose}
The value of problem \eqref{obj1} is finite if and only if $y<\hat{y}$, where $\hat{y}$ is given by \eqref{define:haty}.
 Moreover, when $y<\hat{y}$, the optimal solution of problem \eqref{obj1}, if it exists, must belong to $\sets_y$, where $\sets_y$ is given by \eqref{define:sety}.
\end{thm}
If $y\geqslant \hat{y}$ or $\hat{y}= 0$, then the value of problem \eqref{obj1} and that of the original problem \eqref{obj0} are both infinity. That is to say, if the benchmark level $\ell$ compared to the initial endowment $x$ is too aggressive, then the investment risk will be beyond \emph{any} agent's tolerance. By contrast, in the classic mean-variance model, \emph{some} agents may enter the market regardless of the
benchmark level because the investment risk will still be lower than his/her tolerance.
\par
In the following part, we look for the optimal solution of problem \eqref{obj1}. From now on, we assume $0<y<\hat{y}$.

\subsection{Optimal Solution}
\noindent
The biggest difficulty in solving problem \eqref{obj1} is to overcome the nonlinearity of the {\dindex}. To overcome this, we introduce a series of problems and study their relationships. Eventually, we link problem \eqref{obj1} to a solvable classic portfolio selection problem, and then deduce its optimal solution and optimal value. It will be seen that {\blemma} plays a key role in this approach.
\par
Before introducing a new series of problems, we firstly need to analyze the value function of problem \eqref{obj1}.
\begin{lemma}\label{vconvex}
The value function $V(\cdot)$ of problem \eqref{obj1} is finite, increasing and convex on $[0,\hat{y})$.
\end{lemma} 
By the convexity of the value function, we obtain that
\begin{coro}
The value function $V(\cdot)$ is continuous on $[0,\hat{y})$.\footnote{In fact, $V(\cdot)$ is convex on $(-\infty,\hat{y})$, so $V(\cdot)$ is continuous on $(-\infty,\hat{y})$. In particular, $V(0)=0$.}
\end{coro}
From the above result, it can be seen that if the benchmark level $\ell$ is very close to the initial endowment $x$, then the risk can be arbitrarily small. In other words, if an agent is not greedy, then he/she can set a benchmark level which brings the investment risk within his/her level of tolerance (no matter how small it is, provided it is not zero). That is to say, he/she will enter the market. By contrast, in the classic mean-variance model without a risk-free asset, there is a so-called system risk which is positive. If the system risk is beyond the agent's risk tolerance, then no matter how small the benchmark level is, he/she will not enter the market.
\par
It is clear that the {\dindex} which is the target of problem \eqref{obj1} is not easy to deal with. To avoid this, we introduce a bridge problem
\begin{align}\label{obj2}
 \sup\limits_{(\alpha,Y)} \;\alpha\qquad
 \mathrm{s.t.} \quad \BE[e^{-\alpha u(Y)}]\leqslant 1,\quad \BE[\rho Y]=-y, \quad Y\in \lone.
\end{align}
Our novel approach is based on the following result which suggests the relationships between the above problem and problem \eqref{obj1}.
\begin{prop}\label{obj1=obj2}
Fix $y\in(0,\hat{y})$.
A pair $(\alpha^*, Y^*)$ is an optimal solution to problem \eqref{obj2} if and only if $Y^*$ is an optimal solution to problem \eqref{obj1} and $0<\alpha^*=1/\Rin(u(Y^*))<+\infty$.
\end{prop} 
The above result links problem \eqref{obj1} to problem \eqref{obj2}. We now only need to consider the latter.
Problem \eqref{obj2} is significantly simpler than problem \eqref{obj1} since the {\dindex} has been removed.
However, there are still two variables that need to be optimized in problem \eqref{obj2}, and it may not be easy to find the optimal solution directly because the feasible set may not be convex and the standard Lagrange method cannot be applied directly. Our second novel idea is therefore to introduce a new series of single-variable optimization problems which are easier to solve but still closely linked with the bridge problem \eqref{obj2} and therefore also with problem \eqref{obj1}.
\par
Fix $y\in(0,\hat{y})$. Define a series of single-variable optimization problems parameterized by $\alpha>0$.
\begin{align}\label{obj3}
 \PP_{\alpha}:\quad \inf\limits_{ Y} \; \BE[e^{-\alpha u(Y)}] \qquad \mathrm{s.t.} \quad \BE[\rho Y]=-y, \quad Y\in \lone .
\end{align}
This is a Merton's optimum consumption-like problem, in which the value of absolute risk aversion $\alpha$ is given as a  priori.

Denote the value of the above problem as $\vp(\alpha)$ which is clearly nonnegative and finite, since $Y=-y$ is a feasible solution of it with a finite value. For the convenience of notation, $\vp(0)$ is naturally defined as 1.
\par
Before studying problem $ \PP_{\alpha}$, we need some preparations. For every $\alpha>0$, the mapping $x\mapsto e^{-\alpha u(x)}$ is strictly convex and decreasing on $\R$, so its derivative $x\mapsto -\alpha u'(x)e^{-\alpha u(x)}$ is a one to one increasing mapping from $\R$ to $(-\infty, 0)$. Let $I_{\alpha}(\cdot)$ denote the inverse of the mapping $x\mapsto -\alpha u'(x)e^{-\alpha u(x)}$. Then $I_{\alpha}(\cdot)$ is a one to one increasing mapping from $(-\infty, 0)$ to $\R$.
\begin{assump}\label{assump:I}
For each $\alpha>0$,
 $ \rho I_{\alpha}\left(\lambda \rho\right)\in\lone $ and $ I_{\alpha}\left(\lambda \rho\right) \in\lone $, whenever $\lambda\in(-\infty, 0)$.\footnote{This assumption may be replaced by a weaker one. }
\end{assump}
The behaviors of $\BE[\rho I_{\alpha}\left(\lambda \rho\right)]$ and $\BE[ I_{\alpha}\left(\lambda \rho\right)]$ are very complicated. We refer interested readers to Jin, Xu and Zhou (2008) for a complete study of a similar problem.
\begin{example}
If $u(x)=x.$ Then $I_{\alpha}(x)=- \alpha^{-1}\ln(-\alpha^{-1} x).$ Assumption \ref{assump:I} is satisfied.
\end{example}
Now we turn to solve problem $ \PP_{\alpha}$. The complete solution is given below.
\begin{prop}\label{obj3solution}
Fix $y\in(0,\hat{y})$. For each $\alpha>0$, problem $\PP_{\alpha}$ admits a unique solution
\begin{align}\label{y-alpha}
 Y_{\alpha}=I_{\alpha}\left(\lambda \rho\right),
\end{align}
where $\lambda\in (-\infty, 0)$ is the root of
\begin{align*}
\BE[\rho I_{\alpha}\left(\lambda \rho\right)]=-y.
\end{align*}
Moreover, the optimal value of problem $\PP_{\alpha}$ is given by $\vp(\alpha)=\BE[e^{-\alpha u(Y_{\alpha})}]$.
\end{prop} 
Although problem $\PP_{\alpha}$ is completely solved in the above proposition, its relationship to problem \eqref{obj2} needs to be ascertained. \par
To answer this question, we need to study the properties of $\vp(\cdot)$. The first one is the continuity property of $\vp(\cdot)$ which will be proved through its convexity.
\begin{lemma}\label{v2convex}
 Fix $y\in(0,\hat{y})$.
The value function $\vp(\cdot)$ is convex on $[0,\infty)$ .
\end{lemma} 
\begin{coro}
The function $\vp (\cdot)$ is nonnegative, finite and continuous on $(0,\infty)$.
\end{coro}
The next result is dedicated to showing that the value function $\vp(\cdot)$ is strictly less than 1 near 0.
\begin{lemma} \label{vp<1}
 Fix $y\in(0,\hat{y})$.
There is $ \hat{\alpha}>0$ such that $\vp (\alpha)<1$ whenever $0<\alpha<\hat{\alpha}$.
\end{lemma} 
From the above lemma, $\vp (0+)\leqslant 1=\vp (0)$. It is natural to ask whether $\vp (\cdot)$ is continuous at $0$. The answer is not true in general. This is very different from the answer to the value function $V(\cdot)$ of problem \eqref{obj1}. See the example in the following section.
\par
We have the following important byproduct.
\begin{coro}\label{v2=1uniqe}
The equation $\vp (\alpha)=1$ admits at most one positive solution.
\end{coro} 
The following result shows that the value function $\vp(\cdot)$ can be arbitrary large.
\begin{lemma}\label{vplowerb}
For each $\alpha>0$, $\vp (\alpha)\geqslant e^{\alpha u'(0)y-\BE[\rho\ln(\rho)]}.$
\end{lemma} 
Since $\vp(\alpha)$ is less than 1 when $\alpha$ is close to zero and bigger than 1 when $\alpha>\frac{yu'(0)}{\BE[\rho\ln(\rho)] }$, so by the concavity, we conclude that
\begin{coro}
The equation $\vp(\alpha)=1$ admits a unique positive solution.
\end{coro}
Now we are ready to reveal the relationship between problem $\PP_{\alpha}$ and problem \eqref{obj2}.
\begin{prop}\label{v2->v}
Fix $y\in(0,\hat{y})$.
Suppose $\vp (\alpha^*)=1$ for some $\alpha^*>0$. Let $Y^*$ be the optimal solution to problem $\PP_{\alpha^*}$. Then the pair $(\alpha^*, Y^*)$ is an optimal solution to problem \eqref{obj2}.
\end{prop}
\begin{proof}
Since $\vp (\alpha^*)=1$, $(\alpha^*, Y^*)$ is a feasible solution of problem \eqref{obj2}. Suppose it is not optimal.
Then there is a pair $(\alpha , Y )$ satisfying $\BE[e^{-\alpha u(Y )}]\leqslant 1$, $\BE[\rho Y ]=-y$, $Y \in \lone$, and $\alpha >\alpha^*>0$.
This $Y $ is a feasible solution of problem $\PP_{\alpha }$, so $\vp (\alpha )\leqslant\BE[e^{-\alpha u(Y )}] \leqslant 1.$
On the other hand, by the convexity of $\vp (\cdot)$,
$$ 1=\vp (\alpha^*)\leqslant \frac{\alpha -\alpha^*}{\alpha }\vp (0+)+\frac{ \alpha^*}{\alpha }\vp (\alpha )\leqslant \frac{\alpha -\alpha^*}{\alpha }+\frac{ \alpha^*}{\alpha }\vp (\alpha ),\quad 1\leqslant \vp (\alpha ).$$
So $\vp (\alpha )=1$. That $\vp (\alpha )=\vp (\alpha^*)=1$ contradicts the claim of Corollary \ref{v2=1uniqe}.
\end{proof}
Putting all of the results obtained thus far together, we obtain the complete solution of problem \eqref{obj1}.
\begin{thm} Fix $y\in(0,\hat{y})$.
Then $\vp (\alpha)=\BE[e^{-\alpha u(Y_{\alpha})}]=1$ admits a unique positive solution
$\alpha^*>0$, where $Y_{\alpha}$ is given by \eqref{y-alpha}. Moreover, $(\alpha^*, Y_{\alpha^*})$ is an optimal solution to problem \eqref{obj2}, and $Y_{\alpha^*}$ is an optimal solution to problem \eqref{obj1} with the optimal value $1/\alpha^*$, and $0<\alpha^*\leqslant \frac{\BE[\rho\ln(\rho)] }{yu'(0)}$.
\end{thm}
\par
The optimal value of problem \eqref{obj1} is lower bounded by $\frac{yu'(0)}{\BE[\rho\ln(\rho)] }$, which is increasing in $y$. According to Lemma \ref{vconvex}, the bigger the surplus level, the higher the investment risk; the growth rate by the above result is also at least $\frac{u'(0)}{\BE[\rho\ln(\rho)] }$. If the surplus level is too large, or to be more precise, is larger than $\hat{y}$, then no risk-averse agent will enter the market.

\subsection{Examples}
\noindent
In this section, we present two important examples to illustrate the main result of this paper.
\begin{example}
Risk-neutral agent. In this case, $u(x)=x$ and $\hat{y}=+\infty$.\footnote{Here we assume that $\mathrm{essinf}\:\rho=0$.} Then $I_{\alpha}(x)=- \frac{1}{\alpha}\ln(-\frac{x}{\alpha})$. The optimal solution to problem $\PP_{\alpha }$ is $I_{\alpha}( \lambda \rho )= -y+\frac{1}{\alpha}(\BE[\rho\ln(\rho)]-\ln(\rho))$ and optimal value is $\vp(\alpha)=\BE[e^{-\alpha u(I_{\alpha}( \lambda \rho ))}]=e^{\alpha y-\BE[\rho\ln(\rho)]}$. From $\vp (\alpha^*)=1$, we obtain $\alpha^*=\BE[\rho\ln(\rho)]/y>0$. The optimal solution to problem \eqref{obj1} is $\frac{-\ln(\rho)y}{\BE[\rho\ln(\rho)]}$ and the optimal value of problem \eqref{obj1} is $V(y)=1/\alpha^*=\frac{y}{\BE[\rho\ln(\rho)]}$, which is continuous on $[0,\infty)$. However, the value function $\vp (\cdot)$ is not continuous at 0 since $\vp (0+)=e^{-\BE[\rho\ln(\rho)]}<1=\vp (0)$.
In this case, the investment risk is proportional to the surplus level.
\end{example} 
 \begin{example}
 Risk-averse agent. Let $u(x)=1-e^{-\beta x}$, $\beta>0$ and $\ln(\rho)$ be normally distributed with mean $\mu$ and variance $\sigma^2$. Then we deduce $\mu=-\frac{1}{2}\sigma^2$ from $\BE[\rho]=1$. It is not hard to derive  
 $\hat{y}=\frac{\sigma^2 }{2\beta}$.
 In this case, $I_{\alpha}(\cdot)$ is the unique function satisfying $\alpha e^{-\beta I_{\alpha}(x)}-\beta I_{\alpha}(x)=\alpha+\ln(-x)-\ln(\alpha)-\ln(\beta).$
 Let $W(\cdot)$ be the Lambert function, which is the unique function satisfying $ W(x)e^{W(x)}=x$ on $[0,+\infty)$. Then $I_{\alpha}(x)=\frac{1}{\beta}(
 W(-\frac{x}{\beta}e^{\alpha})-\ln(-\frac{x}{\beta}e^{\alpha})+\ln(\alpha))$.
 The following picture shows the relationship between the risk and surplus level.
 \begin{itemize}
\item[] \includegraphics[width=0.9\textwidth]{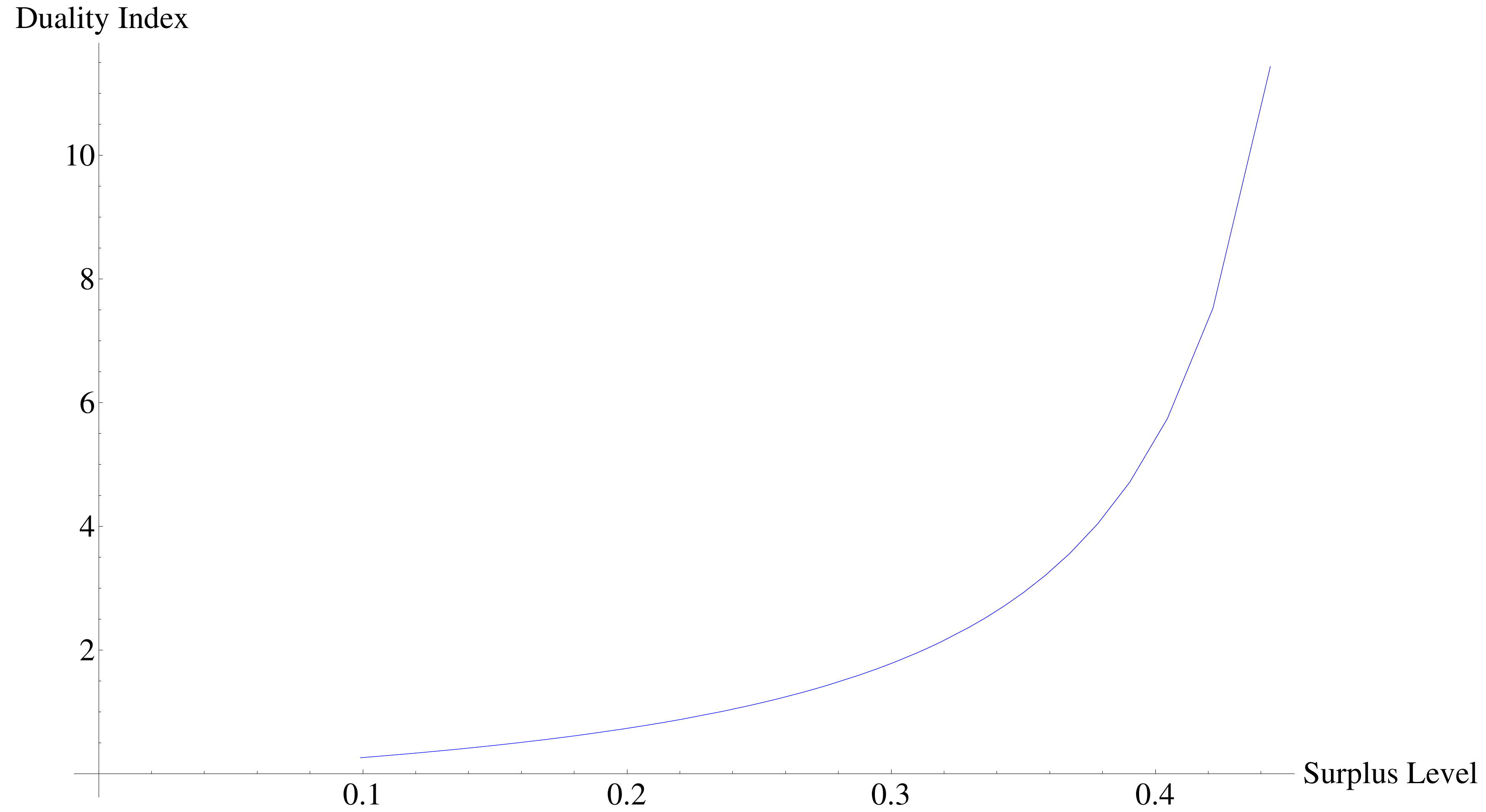}
\end{itemize}
Where $\sigma^2=\beta=1$.
 From the picture, it can be seen that the risk will go to infinity as the surplus level $y$ goes to $\hat{y}=\frac{\sigma^2 }{2\beta}=0.5$, and to 0 as $y$ to 0.
 \end{example}
 \begin{remark}
 The more risk-averse (i.e., the bigger the $\beta$) the agent , the less (i.e., the smaller the $\hat{y}$) his/her risk tolerance.
 \end{remark}

\section{Concluding Remarks}
\noindent
In this paper,   a nontrivial, nonnegative, real-valued index for general outcomes is defined. It is a generalization of the index introduced by Aumann and Serrano (2008). It retains many properties of the original index including sub-additivity, law-invariance, convexity and monotonicity as well as continuity.  It is also the unique nontrivial index that satisfies both the duality and positive homogeneity axioms.
\par
A portfolio selection problem is then considered in a complete market setting. The problem is completely solved by  linking it to a series of Merton's optimum consumption-like problems via a bridge problem. The problem is equivalent to one of the series, which is revealed by an explicit algebraic equation.
\par
In an incomplete market setting, the problem will be much harder to tackle and a new method   needs to be adopted.
Optimal stopping problem under the duality index is also interesting. These problems will be addressed in the forthcoming papers.

\section*{Appendix}
\appendix

\section{Proof of Basic Lemma}
\noindent
By the monotonic convergence theorem, we have $$\lim\limits_{\alpha\to\hat{\alpha}- } \BE[e^{\alpha X^-}\mathbf{1}_{\{X<0\}}]=\BE[e^{\hat{\alpha} X^-}\mathbf{1}_{\{X<0\}}],$$ and by the dominated convergence theorem, we have $$\lim\limits_{\alpha\to\hat{\alpha}- } \BE[e^{-\alpha X^+}\mathbf{1}_{\{X\geqslant 0\}}]=\BE[e^{-\hat{\alpha} X^+}\mathbf{1}_{\{X\geqslant 0\}}].$$
Adding them up, noticing the definition of $\hat{\alpha}$, we obtain
\begin{align}\label{fcontinuous}
1\geqslant \liminf\limits_{\alpha\to\hat{\alpha}- } \BE[e^{-\alpha X } ]=\BE[e^{-\hat{\alpha} X}].
\end{align}
If $\BP(X\neq 0)>0$, set $f(\alpha)=\BE[e^{-\alpha X}]$. Then $f(\cdot)$ on $[0,\hat{\alpha}]$ is a strictly convex function satisfying $f(0)=1$, $f(\hat{\alpha} )\leqslant 1$. Therefore, $f(\alpha)<\frac{\hat{\alpha} -\alpha}{\hat{\alpha} }f(0)+\frac{\alpha}{\hat{\alpha} }f (\hat{\alpha} )\leqslant 1$, whenever $0<\alpha<\hat{\alpha} $.
Since every convex function is continuous in the interior of its domain, we obtain that $f(\cdot)$ is continuous on $(0,\hat{\alpha})$.
By \eqref{fcontinuous}, we see that $f(\cdot)$ is continuous at $\hat{\alpha}$ too.
By the definition of $\hat{\alpha} $, we have that $\BE[e^{-\alpha X}]>1$ whenever $\alpha>\hat{\alpha} $.
\par
If $X\in\seta$, then $\hat{\alpha}>0$ by the definition of $\mgf$. We also have that $\hat{\alpha}<+\infty$. Otherwise $\BE[e^{-\hat{\alpha} X}]\geqslant \BE[e^{+\infty}\mathbf{1}_{\{X< 0\}}]=+\infty$.
\par
If $X\in\setb$, then $\hat{\alpha}=+\infty$ by the definition.
\par
Suppose $\hat{\alpha}>0$ for some $X\in\setc$. Then by the convexity of $f(\cdot)$, we have, for every $\alpha>0$, $\frac{f(\alpha)-f(0)}{\alpha-0}\geqslant f'(0+)=-\BE[X]\geqslant 0$, which contradicts $f(\alpha)=\BE[e^{-\alpha X}]<1$ whenever $0<\alpha<\hat{\alpha} $. We conclude that $\hat{\alpha}=0$ for every $X\in\setc$.
\par
\par
If $X\in\setd$, then $\BE[X^-]=+\infty.$ So $X\notin\mgf$ since all the moments of $X^-$ are finite when $X\in\mgf$.
Therefore, $\BE[e^{-\alpha X } ]=+\infty$ for all $\alpha>0$ and $\hat{\alpha}=0$ by the definition.
\par
 That $f(\cdot)$ is continuous at $0$ whenever $\hat{\alpha}>0$ can be proved by the same idea as proving \eqref{fcontinuous}.

\section{Proof of Proposition \ref{indexproperties}}
\noindent
The proof is very similar to that in Aumann and Serrano (2008). However, our definition of the {\dindex} is different from the original one. So we give the proof here.
\begin{enumerate}
 \item If one of $X$ and $Y$ belongs $\setc\cup\setd$, then $\Rin(X)+\Rin(Y)=+\infty\geqslant \Rin(X+Y)$ (which is defined as $+\infty$ whenever $X+Y\notin\setl$). If one of $X$ and $Y$ belongs to $\setb$, say $Y$. Then $X+Y\geqslant X$ almost surely (a.s.). By the monotonic property which will be proved below, we have $\Rin(X+Y)\leqslant \Rin(X)=\Rin(X)+\Rin(Y).$ Now suppose both $X$ and $Y$ belong to $\seta$. Let $\alpha_1=\sup\{\alpha\geqslant 0: \BE[e^{-\alpha X}]\leqslant 1\}$ and $\alpha_2=\sup\{\alpha\geqslant 0: \BE[e^{-\alpha Y}]\leqslant 1\}$. Then $0<\alpha_1, \alpha_2<+\infty$. Set $k=\frac{\alpha_2}{\alpha_1 +\alpha_2 }\in(0,1)$. Then $\frac{\alpha_1\alpha_2}{\alpha_1+\alpha_2}(X+Y)=k\alpha_1 X+(1-k)\alpha_2 Y.$ By the convexity of the exponential function, we have $\BE[e^{-\frac{\alpha_1\alpha_2}{\alpha_1+\alpha_2}(X+Y)}]=\BE[e^{-k\alpha_1 X-(1-k)\alpha_2 Y}]\leqslant k\BE[e^{-\alpha_1 X}]+(1-k)\BE[e^{-\alpha_2Y}]\leqslant 1.$ By the definition of the {\dindex}, $\Rin(X+Y)\leqslant \frac{\alpha_1+\alpha_2}{\alpha_1\alpha_2}= {\alpha_1^{-1}+\alpha_2^{-1}}=\Rin(X)+\Rin(Y)$.
 \item This follows immediately from the definition.
 \item This follows from the above two properties.
 \item This is evident.
 \item Because the {\dindex} is law-invariant, it is sufficient to prove the second-order case. Suppose $Y$ second-order stochastic dominates $X$. Then for each $\alpha>0$, the mapping $x\mapsto e^{-\alpha x}$ is decreasing and convex, so $\BE[e^{-\alpha Y}]\leqslant \BE[e^{-\alpha X}]$. Therefore, $\sup\{\alpha\geqslant 0: \BE[e^{-\alpha X}]\leqslant 1\}\leqslant \sup\{\alpha\geqslant 0: \BE[e^{-\alpha Y}]\leqslant 1\}$. By the definition of the {\dindex}, we have $\Rin(X)\geqslant \Rin(Y)$.
\end{enumerate}

\section{Proof of Lemma \ref{lemma:chara}}
\noindent
For each $\alpha\geqslant 0$, by the monotonic convergence theorem, we have $\lim\limits_{n\to+\infty } \BE[e^{\alpha X_{n}^-}\mathbf{1}_{\{X\leqslant0\}}]=\BE[e^{ {\alpha} X^{-}}\mathbf{1}_{\{X\leqslant0\}}]$ and $\lim\limits_{n\to+\infty } \BE[e^{-\alpha X_{n}^+}\mathbf{1}_{\{X>0\}}]=\BE[e^{ -{\alpha} X^{+}}\mathbf{1}_{\{X>0\}}]$. Adding them up, we obtain that $\lim\limits_{n\to+\infty } \BE[e^{-\alpha X_n } ]=\BE[e^{ -{\alpha} X} ]$. This confirms that $\lim\limits_{n\to +\infty} \Rin(X_n)= \Rin(X)$.
\par
Similarly, by the monotonic convergence theorem, we have $\lim\limits_{n\to+\infty } \BE[u(w+X_{n}^+)\mathbf{1}_{\{X>0\}}]=\BE[u(w+X)\mathbf{1}_{\{X>0\}}]$ and
$\lim\limits_{n\to+\infty } \BE[u(w-X_{n}^-)\mathbf{1}_{\{X\leqslant 0\}}]=\BE[u(w+X)\mathbf{1}_{\{X\leqslant 0\}}]$. Adding them up, we obtain that $\lim\limits_{n\to+\infty } \BE[u(w+X_n)]=\BE[u(w+X)]$.

\section{Proof of Theorem \ref{indexunique}}
\noindent
Since the {\dindex} defined in Definition \ref{define:uindex} coincides with that in Aumann and Serrano (2008) for all bounded outcomes, it also satisfies the duality axiom for all bounded outcomes.
\par
 Assume that the utility functions of two agents $A$ and $B$ are $u_A(\cdot)$ and $u_B(\cdot)$, respectively, and the agent $A$ is uniformly more risk-averse than agent $B$. Suppose that the agent $A$ accepts an outcome $X\in\setl$ at some wealth $w_0$, i.e., either $X=0$ a.s. or $\BE[u_A(w_0+X)]>u_A(w_0).$ We need to prove that the agent $B$ accepts any outcome $Y\in\setl$ satisfying $\Rin(Y)<\Rin(X)$ at any wealth $w$, that is to say, either $Y=0$ a.s. or $\BE[u_B(w+Y)]>u_B(w)$.
\par
Since $\Rin(Y)<\Rin(X)\leqslant +\infty$, $Y\in\seta\cup\setb$. If $Y=0$, then it is accepted. If $Y\in\setb-\{0\}$, then $Y\geqslant 0$ a.s., $\BE[u_B(w+Y)]>u_B(w)$ and it is accepted too.
 \par
Now fix $Y\in\seta$ and $w\in\R$. Define $Y_{\varepsilon}=Y-\varepsilon \mathbf{1}_{\{Y>0\}}$, for every $\varepsilon>0$. Then for each $\alpha>0$ such that $\BE[e^{-\alpha Y}]<1$, we have $\lim\limits_{\varepsilon \to 0+}\BE[e^{-\alpha Y_{\varepsilon}}]=\BE[e^{-\alpha Y}]<1$. So we conclude that $\lim\limits_{\varepsilon \to 0+} \Rin(Y_{\varepsilon})\leqslant \Rin(Y)<\Rin(X)$. Let us fix an $\varepsilon>0$ such that $\Rin(Y_{\varepsilon})<\Rin(X)$.
 \par
Let $X_n=\min\{\max\{X,-n\},n\}$ and $Y_n=\min\{\max\{Y_{\varepsilon},-n\},n\}$. Then by Lemma \ref{lemma:chara}, $\lim\limits_{n\to+\infty } \Rin(Y_n)=\Rin(Y_{\varepsilon})<\Rin(X)=\lim\limits_{n\to+\infty } \Rin(X_n)$. So $\Rin(Y_n)<\Rin(X_n)$ for every large $n$. On the other hand, by Lemma \ref{lemma:chara}, $\lim\limits_{n\to+\infty } \BE[u_A(w_0+X_n)]=\BE[u_A(w_0+X)]>u(w_0)$, so the agent $A$ accepts $X_n$ at wealth $w_0$ for every large $n$. Both $X_n$ and $Y_n$ are bounded, so the agent $B$ accepts $Y_n$ at any wealth $w$, that is $\BE[u_B(w+Y_n)]>u_B(w)$ for every large $n$. Applying Lemma \ref{lemma:chara} again, $\BE[u_B(w+Y_{\varepsilon})]=\lim\limits_{n\to+\infty } \BE[u_B(w+Y_n)]\geqslant u_B(w)$. Since $Y\in\seta$, $\BP(Y>0)>0$, $\BE[u_B(w+Y)]>\BE[u_B(w+Y-\varepsilon \mathbf{1}_{\{Y>0\}})]=\BE[u_B(w+Y_{\varepsilon})]\geqslant u_B(w)$. That is to say, the agent $B$ accepts $Y$ at wealth $w$. Now we proved that the {\dindex} satisfies the duality axiom.
\par
It is evident that the {\dindex} satisfies the positive homogeneity axiom.
\par
The uniqueness of the index can be proved by a similar limit argument. We leave the proof to interested readers.

\section{Proof of Lemma \ref{setdemp}}
\noindent
Suppose $\hat{\sets}_y$ is not empty, then there is $Y\in \lone$ satisfying $\BE[\rho Y]=-y$ and $\BE[u(Y)]>0 $. Then for any $y'<y$, we have $Y+y-y'\in \lone$, $\BE[\rho (Y+y-y')]=-y'$, and $\BE[u(Y+y-y')]\geqslant \BE[u(Y)]>0 $. This indicates that the set $\hat{\sets}_{y'}$ is not empty. By the definition of $\hat{y}$, we conclude that the set $\hat{\sets}_y$ is not empty whenever $y<\hat{y}$.
\par
 It is evident that the set $\hat{\sets}_y$ is empty whenever $y>\hat{y}$ by the definition of $\hat{y}$.
 \par
 We only need to prove that $\hat{\sets}_{\hat{y}}$ is empty whenever $\hat{y}<+\infty$. Suppose not, then there is $Y\in \lone$ satisfying $\BE[\rho Y]=-\hat{y}$ and $\BE[u(Y)]>0 $. Let $Y_{\varepsilon}=Y-\varepsilon\mathbf{1}_{\{Y>0\}}$. By the monotonic convergence theorem,
$ \lim\limits_{\varepsilon\to 0+}\BE[u(Y_{\varepsilon} )\mathbf{1}_{\{Y>0\}}]=\BE[u(Y )\mathbf{1}_{\{Y>0\}}]\geqslant 0$. It is evident that $\BE[u(Y_{\varepsilon} )\mathbf{1}_{\{Y\leqslant 0\}}]=\BE[u(Y )\mathbf{1}_{\{Y\leqslant 0\}}] >-\infty$ since $\BE[u(Y)]>0 $. Adding them up, we obtain $ \lim\limits_{\varepsilon\to 0+}\BE[u(Y_{\varepsilon} ) ]=\BE[u(Y ) ]>0$. Thus, there is $\varepsilon>0$ such that $\BE[u(Y_{\varepsilon} ) ]>0$.
Set $\delta=\varepsilon\BE[\rho\mathbf{1}_{\{Y>0\}}]\geqslant 0$. If $\delta=0$, then $Y\leqslant 0$ a.s. and $\BE[u(Y)]\leqslant 0 $, a contradiction. So $\delta>0$. Since $\BE[\rho Y_{\varepsilon}]=\BE[\rho Y]-\varepsilon\BE[\rho\mathbf{1}_{\{Y>0\}}]=-\hat{y}-\delta<-\hat{y}$, we have that $Y_{\varepsilon}\in\hat{\sets}_{\hat{y}+\delta}$ which contradicts the definition of $\hat{y}$. The proof is complete.

\section{Proof of Lemma \ref{setdemp2}}
\noindent
 Suppose $y<\hat{y}$. Let $\varepsilon>0$ such that $y+\varepsilon<\hat{y}$. Then by Lemma \ref{setdemp}, $ \hat{\sets}_{y+\varepsilon}$ is not empty, so there is $Y\in \lone$ satisfying $\BE[\rho Y]=- (y+\varepsilon)$ and $\BE[u(Y)]>0$. Let $Y_n=\max\{Y,-n\}$. Since $Y\leqslant Y_n\leqslant Y^+$, by the monotonic convergence theorem,
 $\lim\limits_{n\to +\infty} \BE[\rho Y_n] =\BE[\rho Y]=- (y+\varepsilon)<-y$ and $\lim\limits_{n\to +\infty} \BE[u(Y_n)]=\BE[u(Y)]>0$. Therefore, we have $\BE[\rho Y_n] <-y$ and $\BE[u(Y_n)]>0$ for large $n$. Let $\delta>0$ satisfy $\BE[\rho (Y_n+\delta)] =-y$. It is very easy to verify that $ Y_n+\delta $ belongs to the set $ \sets_y$. So the set $ \sets_y$ is not empty.
 \par
 For every $y\geqslant\hat{y}$, by Lemma \ref{setdemp}, $\hat{\sets}_y$ is empty. Because the set $\sets_y$ is a subset of $\hat{\sets}_y$, so it is empty as well.

\section{Proof of Lemma \ref{vconvex}}
\noindent
Let $y_1<y_2<\hat{y}$.
For any two outcomes $Y_1, Y_2\in \lone $ satisfying $\BE[\rho Y_1]=-y_1$, $\BE[\rho Y_2]=-y_2$ and any constant $k \in(0,1)$, by the concavity of $u(\cdot)$,
\begin{align*}
u(k Y_1+(1-k )Y_2)\geqslant k u( Y_1)+(1-k )u(Y_2).
\end{align*}
Because the {\dindex} is monotonically decreasing w.r.t. the first-order dominance and convex,
\begin{align*}
\Rin(u(k Y_1+(1-k )Y_2))&\leqslant \Rin(k u( Y_1)+(1-k )u(Y_2))
 \leqslant k \Rin( u( Y_1))+(1-k )\Rin(u(Y_2)),
\end{align*}
which implies that
\begin{align*}
V(k y_1+(1-k )y_2)\leqslant kV(y_1)+(1-k )V(y_2).
\end{align*}
The monotonically increasing of $V(\cdot)$ is due to the fact that the {\dindex} is monotonically decreasing w.r.t. the first-order dominance. That $V(\cdot)$ is finite on $[0,\hat{y})$ is due to Theorem \ref{wellpose}.

\section{Proof of Proposition \ref{obj1=obj2}}
\noindent
``$\Longrightarrow$'':
We first prove that $0<\alpha^*<+\infty$, then by the definition of the {\dindex}, $\alpha^*=1/\Rin(u(Y^*))$ follows immediately. Since $y\in(0,\hat{y})$, the set $\sets_y$ is not empty, so there is $Y\in \lone$ satisfying $\BE[\rho Y]=-y$ and $u(Y)\in\seta$. Then $0<\Rin(u(Y))<+\infty$. Since $(1/\Rin(u(Y)),Y)$ is a feasible solution of problem \eqref{obj2}, $\alpha^*\geqslant 1/\Rin(u(Y))>0$.
On the other hand, if $\alpha^*=+\infty$, then $1\geqslant \BE[e^{-\alpha^* u(Y^*)}]\geqslant +\infty*\BP(Y^*<0)$, so $\BP(Y^*<0)=0$, $\BE[\rho Y^*]\geqslant 0>-y$, a contradiction.
\par
Next we prove that $Y^*$ is an optimal solution to problem \eqref{obj1}.
 Otherwise, there is $Y\in \lone$ satisfying $\BE[\rho Y]=-y$ and $\Rin(u(Y))<\Rin(u(Y^*))$. Then $u(Y)\in\seta\cup\setb$. Since
 $u(Y)\in\setb$ implies that $Y\geqslant 0$ a.s. and $\BE[\rho Y]\geqslant 0>-y$, we conclude that $u(Y)\in\seta$ and so $0<\Rin(u(Y))<+\infty$. Then the pair $(1/\Rin(u(Y)),Y)$ is a feasible solution of problem \eqref{obj2} and $1/\Rin(u(Y))>1/\Rin(u(Y^*))=\alpha^*$, which contradicts the optimality of $(\alpha^*, Y^*)$ to problem \eqref{obj2}.
\par
``$\Longleftarrow$'': Suppose $Y^*$ is an optimal solution to problem \eqref{obj1}. Since $y\in(0,\hat{y})$, the optimal value $\Rin(u(Y^*))<+\infty$. Since $\Rin(u(Y^*))=0$ leads to $Y^*\geqslant 0$ a.s. and $\BE[\rho Y^*]\geqslant 0>-y$, a contradiction, we conclude that $0<\Rin(u(Y^*))<+\infty$.
Then $(1/\Rin(u(Y^*)), Y^*)$ is a feasible solution to problem \eqref{obj2}. Suppose it is not optimal. Then there is a pair $(\alpha, Y) $ satisfying $Y\in \lone$, $ \BE[e^{-\alpha u(Y)}]\leqslant 1$, $\BE[\rho Y]=-y$ and $\alpha>1/\Rin(u(Y^*))> 0$. Then $Y$ is a feasible solution of problem \eqref{obj1}, but $\Rin(u(Y)) \leqslant 1/\alpha<\Rin(u(Y^*))$, which contradicts the optimality of $Y^*$ to problem \eqref{obj1}.

\section{Proof of Proposition \ref{obj3solution}}
\noindent
By the monotonic convergence theorem, the mapping $\lambda\mapsto\BE[\rho I_{\alpha}\left(\lambda \rho\right)]$ is continuous and increasing on $(-\infty, 0)$ and
\begin{align*}
\lim\limits_{\lambda\to -\infty }\BE[\rho I_{\alpha}\left(\lambda \rho\right)]&=\BE[\lim\limits_{\lambda\to -\infty }\rho I_{\alpha}\left(\lambda \rho\right)]=-\infty,\\
\lim\limits_{\lambda\to 0-}\BE[\rho I_{\alpha}\left(\lambda \rho\right)]&=\BE[\lim\limits_{\lambda\to 0-}\rho I_{\alpha}\left(\lambda \rho\right)]=+\infty,
\end{align*}
so $\BE[\rho I_{\alpha}\left(\lambda \rho\right)]=-y$ admits a negative solution.
\par
The optimality of \eqref{y-alpha} can be shown by the standard Lagrange method. We leave the proof to interested readers.

\section{Proof of Proposition Lemma \ref{v2convex}}
\noindent
Let $\alpha>\alpha'\geqslant 0$ be two scalars, and $Y$, $Y'$ be the corresponding feasible solutions of problem $\PP_{\alpha}$ and $\PP_{\alpha'}$, respectively. Then for any $ k \in(0,1)$,
by the convexity, monotonicity of exponential function, and the concavity of $u(\cdot)$,
\begin{multline*}
k \BE\left[ e^{-\alpha u(Y)}\right]+(1- k ) \BE\left[e^{-\alpha' u(Y')}\right]= \BE\left[ k e^{-\alpha u(Y)}+(1- k )e^{-\alpha' u(Y')}\right]\\
\geqslant \BE\left[e^{- k \alpha u(Y) -(1- k )\alpha' u(Y')}\right]
=\BE\left[e^{-( k \alpha+(1- k )\alpha' )(\beta u(Y)+(1-\beta ) u(Y'))}\right]\\
\geqslant \BE\left[ e^{-( k \alpha+(1- k )\alpha' )( u(\beta Y +(1-\beta) Y'))}\right]\geqslant \vp ( k \alpha+(1- k )\alpha' ),
 \end{multline*}
where $\beta=\frac{ k \alpha}{ k \alpha+(1- k )\alpha'}\in(0,1]$.
The convexity of $\vp (\cdot)$ on $[0,\infty)$ is established.

\section{Proof of Lemma \ref{vp<1}}
\noindent
By Proposition \ref{vfinite}, $\sets_y$ is not empty, so there is $Y\in \lone $ satisfying $\BE[\rho Y]=-y$, and $u(Y)\in\seta$. By {\blemma}, there is $\hat{\alpha}>0$ such that $\vp (\alpha)\leqslant \BE[e^{-\alpha u(Y)}]<1$ whenever $0<\alpha<\hat{\alpha}$.

\section{Proof of Corollary \ref{v2=1uniqe} }
\noindent
Suppose that $\vp (\alpha')=\vp (\alpha'')=1$ for some $\alpha''>\alpha'>0$. By the Lemma \ref{vp<1}, there is $0<\alpha<\alpha'$ such that $\vp (\alpha )<1$. By the convexity of $\vp (\cdot)$, we deduce that
$$ 1=\vp (\alpha')\leqslant \frac{\alpha''-\alpha'}{\alpha''-\alpha}\vp (\alpha)+\frac{\alpha'-\alpha}{\alpha''-\alpha}\vp (\alpha'')<1.$$
This confirms our claim.

\section{Proof of Lemma \ref{vplowerb} }
\noindent
We first consider the following problem
\begin{align*}
 \inf\limits_{ Y} \; \BE[e^{-\alpha u'(0)Y}] \qquad \mathrm{s.t.} \quad \BE[\rho Y]=-y.
\end{align*}
The standard Lagrange method gives the optimal solution $Y^*=-y+\frac{1}{\alpha u'(0)}(\BE[\rho\ln(\rho)]-\ln(\rho)),$
and optimal value $\BE[e^{-\alpha u'(0)Y^*}]=e^{\alpha u'(0)y-\BE[\rho\ln(\rho)]}.$
\par
Because $u(\cdot)$ is concave and $u(0)=0$, we have that $u(x)\leqslant u'(0)x$ for all $x\in\R$.
Therefore, for each $\alpha>0$ and $Y\in \lone $ satisfying $\BE[\rho Y]=-y$, we have
\begin{align*}
 \BE[e^{-\alpha u(Y)}] \geqslant \BE[e^{-\alpha u'(0)Y}]\geqslant \BE[e^{-\alpha u'(0)Y^*}]=e^{\alpha u'(0)y-\BE[\rho\ln(\rho)]}.
\end{align*}
The claim follows immediately.




\begin{thebibliography}{99}

\bibitem{aumann&serrano}	\textsc{Aumann, R. J., and R. Serrano (2008):}
An Economic Index of Riskness,
\textit{Journal of Political Economy}, Vol.116, pp. 810-836

\bibitem{DS} \textsc{Diamond, P. A., and J. E. Stiglitz (1974):}
Increases in Risk and Risk Aversion,
\textit{Journal of Economic Theory}, Vol.8, pp. 337-360


\bibitem{JXZ} \textsc{Jin, H., Z. Q. Xu and X. Y. Zhou (2008):}
A Convex Stochastic Optimization Problem Arising from Portfolio Selection,
\textit{Journal of Mathematical Finance}, Vol.18, pp. 171-183

\bibitem{Kahneman&Tversky 1979}\textsc{Kahneman, D., and A. Tversky (1979):} Prospect Theory: An Analysis of Decision Under Risk,
\textit{Econometrica}, Vol. 46, pp. 171-185


\bibitem{MM}\textsc{Machina, M., and M. Rothschild (2008):}
Risk,
in \textit{The New Palgrave Dictionary of Economics},
2nd edition, ed. by S. N. Durlauf and L. E. Blume

\bibitem{Merton}\textsc{Merton, R. C. (1971):}
Optimum Consumption and Portfolio Rules in A Continuous-Time Model,
\textit{Journal of Economic Theory}, Vol.3, pp. 373-413


\bibitem{Pardoux&peng}	\textsc{Pardoux, E, and S. G. Peng (1990):}
Adapted Solution of A Backward Stochastic Differential Equation,
\textit{Systems and Control Letters}, Vol.14, pp. 55-61

\bibitem{w&K1992} \textsc{Tversky, A., and D. Kahneman (1992):}
Advances in Prospect Theory: Cumulative Representation of Uncertainty,
\textit{Journal of Risk Uncertainty}, Vol. 5, pp. 297-323


\end{thebibliography}
\end{document}